\documentclass[10pt,reqno]{amsart}
\usepackage[hypertex]{hyperref}

\usepackage{amsmath,amsthm}
\usepackage{amssymb}
\usepackage{euscript}
\usepackage[frame,ps,matrix,arrow,curve,rotate,all,2cell,tips]{xy}

\numberwithin{equation}{subsection}

\newcommand{\sqsp}{\renewcommand{\baselinestretch}{1.1}\tiny\normalsize}

\newtheorem{theorem}[subsection]{Theorem}
\newtheorem{lemma}[subsection]{Lemma}
\newtheorem{proposition}[subsection]{Proposition}
\newtheorem{corollary}[subsection]{Corollary}

\theoremstyle{definition}

\newtheorem{example}[subsection]{Example}

\newcommand{\Ubar}{\overline{U}}

\newcommand{\br}{\mathfrak{B}_n}

\newcommand{\bk}{\mathbf{k}}

\DeclareMathOperator{\Hom}{Hom}
\DeclareMathOperator{\Aut}{Aut}

\begin{document}

\title{Hom-Yang-Baxter equation, Hom-Lie algebras, and quasi-triangular bialgebras}
\author{Donald Yau}

\begin{abstract}
We study a twisted version of the Yang-Baxter Equation, called the Hom-Yang-Baxter Equation (HYBE), which is motivated by Hom-Lie algebras.  Three classes of solutions of the HYBE are constructed, one from Hom-Lie algebras and the others from Drinfeld's (dual) quasi-triangular bialgebras.  Each solution of the HYBE can be extended to operators that satisfy the braid relations.  Assuming an invertibility condition, these operators give a representation of the braid group.
\end{abstract}



\address{Department of Mathematics\\
    The Ohio State University at Newark\\
    1179 University Drive\\
    Newark, OH 43055, USA}
\email{dyau@math.ohio-state.edu}

\date{\today}
\maketitle

\sqsp

\section{Introduction}


The Yang-Baxter Equation (YBE) originated in the work of Yang \cite{yang} and Baxter \cite{baxter72,baxter82} in statistical mechanics.  Let $V$ be a vector space, and let $B \colon V \otimes V \to V \otimes V$ be a linear automorphism.  Then $B$ is said to be an \emph{$R$-matrix} if it satisfies the YBE:
\begin{equation}
\label{eq:YBE}
(Id_V \otimes B) \circ (B \otimes Id_V) \circ (Id_V \otimes B) = (B \otimes Id_V) \circ (Id_V \otimes B) \circ (B \otimes Id_V).
\end{equation}
The YBE has far-reaching mathematical significance.  Indeed, it is closely related to Lie algebras, Drinfeld's (dual) quasi-triangular bialgebras, which include many examples of quantum groups, and the braid group, among other topics.  It is known that every (co)module $M$ over a (dual) quasi-triangular bialgebra $H$ gives a solution of the YBE \cite{dri87,dri89,kassel}.  Also, every Lie algebra $L$ gives a solution of the YBE \cite{bc}.  Moreover, each solution of the YBE gives a representation of the braid group on $n$ strands.


We will study a twisted version of the YBE, which is motivated by Hom-Lie algebras.  A \emph{Hom-Lie algebra} $L$ has a bilinear skew-symmetric bracket $[-,-] \colon L \otimes L \to L$ and a linear map $\alpha \colon L \to L$ such that $\alpha[x,y] = [\alpha(x),\alpha(y)]$ for $x,y \in L$ (multiplicativity) and that the following \emph{Hom-Jacobi identity} holds:
\begin{equation}
\label{eq:hom-jacobi}
[[x,y],\alpha(z)] + [[z,x],\alpha(y)] + [[y,z],\alpha(x)]] = 0.
\end{equation}
A Lie algebra is a Hom-Lie algebra with $\alpha = Id$.  Hom-Lie algebras were introduced in \cite{hls} (without multiplicativity) to describe the structures on certain $q$-deformations of the Witt and the Virasoro algebras.  Earlier precursors of Hom-Lie algebras can be found in \cite{hu,liu}.  Other classes of Hom-Lie algebras were constructed in \cite{ms,yau2}.  We will describe some of these Hom-Lie algebras in Section ~\ref{sec:homlie}.


If one considers a Hom-Lie algebra as an $\alpha$-twisted version of a Lie algebra, then there should be a corresponding twisted YBE.  To state precisely what we mean by a twisted YBE, let $M$ be a vector space with a given linear self-map $\alpha \colon M \to M$, and let $B \colon M \otimes M \to M \otimes M$ be a bilinear map (not necessarily an automorphism) such that $B \circ \alpha^{\otimes 2} = \alpha^{\otimes 2} \circ B$.  We consider the following \emph{Hom-Yang-Baxter Equation} (HYBE) for $(M,\alpha)$,
\begin{equation}
\label{eq:HYBE}
(\alpha \otimes B) \circ (B \otimes \alpha) \circ (\alpha \otimes B) = (B \otimes \alpha) \circ (\alpha \otimes B) \circ (B \otimes \alpha).
\end{equation}
Of course, an $R$-matrix is a solution of the HYBE in which $\alpha = Id$ and $B$ is invertible.  We will construct three classes of solutions of the HYBE, generalizing the solutions of the YBE from Lie algebras and (dual) quasi-triangular bialgebras.


Just as a Lie algebra gives a solution of the YBE \eqref{eq:YBE}, the following result, which will be proved in Section \ref{sec:homlie}, shows that a Hom-Lie algebra gives a solution of the HYBE.  In what follows, $\bk$ denotes the ground field.

\begin{theorem}
\label{thm:homlie}
Let $(L,[-,-],\alpha)$ be a Hom-Lie algebra.  Set $L' = \bk \oplus L$ and $\alpha(a,x) = (a,\alpha(x))$ for $(a,x) \in L'$.  Define a bilinear map $B_\alpha \colon L' \otimes L' \to L' \otimes L'$ by
\begin{equation}
\label{eq:Balpha}
B_\alpha((a,x) \otimes (b,y)) = (b,\alpha(y)) \otimes (a,\alpha(x)) + (1,0) \otimes (0,[x,y]).
\end{equation}
Then $B_\alpha$ is a solution of the HYBE \eqref{eq:HYBE} for $(L',\alpha)$.
\end{theorem}

This Theorem is a generalization of \cite[Proposition 4.2.2]{bc}, which is precisely the case when $\alpha = Id$, i.e., when $L$ is a Lie algebra.

Next we describe solutions of the HYBE from quasi-triangular bialgebras.  A quasi-triangular bialgebra \cite{dri87,dri89} consists of a bialgebra $H$ and an invertible element $R \in H \otimes H$, called the \emph{quasi-triangular structure}.  The comultiplication $\Delta$ in $H$ is almost cocommutative, whose non-cocommutativity is controlled by the element $R$.  Moreover, $R$ satisfies two further compatibility conditions with $\Delta$.  A cocommutative bialgebra is an example of a quasi-triangular bialgebra in which $R = 1 \otimes 1$.  However, most interesting examples of quasi-triangular bialgebras are not cocommutative.  The exact definition of a quasi-triangular bialgebra will be recalled in Section ~\ref{sec:qtb}.

Let $(H,R)$ be a quasi-triangular bialgebra, and let $M$ be an $H$-module.  Define the bilinear map $B_R \colon M \otimes M \to M \otimes M$ by
\begin{equation}
\label{eq:BR}
B_R(u \otimes v) = \tau(R(u \otimes v)),
\end{equation}
where $\tau \colon M \otimes M \to M \otimes M$ is the twist isomorphism.  Then it is well-known that $B_R$ is a solution of the YBE \eqref{eq:YBE} \cite{dri87,dri89,kassel}.  This gives an efficient and systematic way to produce solutions of the YBE.  Particular examples of solutions of the YBE arising this way include the Woronowicz operators on a Hopf algebra \cite{woronowicz}, as shown in \cite{hennings}.

The following generalization will be proved in Section ~\ref{sec:qtb}.

\begin{theorem}
\label{thm:qtb}
Let $(H,R)$ be a quasi-triangular bialgebra, $M$ be an $H$-module, and $\alpha \colon M \to M$ be an $H$-module morphism.  Then the map $B_R$ \eqref{eq:BR} is a solution of the HYBE \eqref{eq:HYBE} for $(M,\alpha)$.
\end{theorem}

Dual to a quasi-triangular bialgebra is the notion of a \emph{dual quasi-triangular bialgebra} \cite{hay,lt,majid91,sch}.  It consists of a bialgebra $H$ and a \emph{dual quasi-triangular structure} $R \in \Hom(H \otimes H, \bk)$.  The exact definition of a dual quasi-triangular bialgebra will be recalled in Section ~\ref{sec:dqtb}.  Let $(H,R)$ be a dual quasi-triangular bialgebra, and let $M$ be an $H$-comodule via the map $\rho \colon M \to H \otimes M$.  For an element $u \in M$, write $\rho(u) = \sum_{(u)} u_H \otimes u_M$.  Define a bilinear map $B^R \colon M \otimes M \to M \otimes M$ by
\begin{equation}
\label{eq:B^R}
B^R(u \otimes v) = \sum_{(u)(v)} R(v_H \otimes u_H) v_M \otimes u_M.
\end{equation}
Then it is well-known that $B^R$ is a solution of the YBE \eqref{eq:YBE}; see, e.g., \cite[Proposition VIII.5.2]{kassel}.  This gives another systematic way to produce solutions of the YBE.  Conversely, by the FRT construction \cite{rtf}, every $R$-matrix for a finite dimensional vector space $M$ has the form $B^R$ for some dual quasi-triangular bialgebra $H$ and some $H$-comodule structure on $M$.  The following result, dual to Theorem ~\ref{thm:qtb}, will be proved in Section ~\ref{sec:dqtb}.

\begin{theorem}
\label{thm:dqtb}
Let $(H,R)$ be a dual quasi-triangular bialgebra, $M$ be an $H$-comodule, and $\alpha \colon M \to M$ be an $H$-comodule morphism.  Then the map $B^R$ \eqref{eq:B^R} is a solution of the HYBE \eqref{eq:HYBE} for $(M,\alpha)$.
\end{theorem}

Solutions of the YBE (i.e., $R$-matrices) can be extended to operators that satisfy the braid relations, which can then be used to construct representations of the braid group.  We extend this construction to solutions of the HYBE.  Let $n \geq 3$ and $\br$ be the braid group on $n$ strands \cite{artin}.  The braid group $\br$ has generators $\sigma_i$ ($1 \leq i \leq n - 1$), which satisfy the determining \emph{braid relations}:
\begin{equation}
\label{eq:braidrelations}
\sigma_i \sigma_j = \sigma_j \sigma_i \text{ if $|i - j| > 1$} \quad \text{ and } \quad
\sigma_i\sigma_{i+1}\sigma_i = \sigma_{i+1}\sigma_i\sigma_{i+1}.
\end{equation}
The following result, which will be proved in Section ~\ref{sec:braid}, shows that each solution of the HYBE can be extended to operators that satisfy the braid relations.  With an additional invertibility condition, these operators give a representation of the braid group $\br$.  It generalizes the usual braid group representations from $R$-matrices, as discussed, for example, in \cite[X.6.2]{kassel}.

\begin{theorem}
\label{thm:braid}
Let $B$ be a solution of the HYBE \eqref{eq:HYBE} for $(M,\alpha)$ and $n \geq 3$.  Define the linear maps $B_i \colon M^{\otimes n} \to M^{\otimes n}$ ($1 \leq i \leq n - 1$) by
\begin{equation}
\label{eq:Bi}
B_i = \begin{cases} B \otimes \alpha^{\otimes (n-2)} & \text{ if $i = 1$},\\
\alpha^{\otimes(i-1)} \otimes B \otimes \alpha^{\otimes (n - i - 1)} & \text{ if $1 < i < n-1$}, \\
\alpha^{\otimes(n-2)} \otimes B & \text{ if $i = n-1$}.\end{cases}
\end{equation}
Then the maps $B_i$ satisfy the braid relations ~\eqref{eq:braidrelations}.  
Moreover, if both $\alpha$ and $B$ are invertible, then so are the $B_i$, and there is a unique group morphism $\rho^B_n \colon \br \to \Aut(M^{\otimes n})$  with $\rho^B_n(\sigma_i) = B_i$.
\end{theorem}

Each Lie algebra gives an $R$-matrix (Theorem ~\ref{thm:homlie} with $\alpha = Id$) $B$ on $L' = \bk \oplus L$.  This in turn gives a corresponding representation of the braid group $\br$ on $L'^{\otimes n}$, as in Theorem ~\ref{thm:braid}.  Our $\alpha$-twisted setting is more flexible and provides many more braid group representations on $L'^{\otimes n}$ for each Lie algebra $L$.  For example, as we will discuss in Example ~\ref{ex:homlie}, each Lie algebra $L$ gives rise to a family $\{L_\alpha = (L,[-,-]_{\alpha},\alpha)\}$ of Hom-Lie algebras, one for each Lie algebra endomorphism $\alpha$ of $L$.  Thus, starting with a Lie algebra $L$ and using Theorems ~\ref{thm:homlie} and ~\ref{thm:braid} and Corollary ~\ref{cor:Balpha} on the sub-family $\{L_\alpha \colon \alpha \text{ invertible}\}$ of Hom-Lie algebras, we obtain a family of representations of $\br$ on $L'^{\otimes n}$.  As an illustration, in Example ~\ref{ex:sl2}, starting with the Lie algebra $sl(2)$, we will construct an explicit infinite, $1$-parameter family of representations of $\br$ on $sl(2)'^{\otimes n} = (\bk \oplus sl(2))^{\otimes n}$.

Likewise, in Theorems ~\ref{thm:qtb} and ~\ref{thm:dqtb} with fixed $(H,R)$ and $M$, suppose we run $\alpha \colon M \to M$ through all the $H$-(co)module automorphisms of $M$.  Then we obtain from Theorem ~\ref{thm:braid} a family, indexed by these $\alpha$, of braid group representations on $M^{\otimes n}$ because $B_R$ \eqref{eq:BR} and $B^R$ \eqref{eq:B^R} are always invertible.

This finishes the descriptions of our main results.  The rest of this paper is organized as follows.  In Section ~\ref{sec:hybe}, we fix some notations and give a few basic examples of solutions of the HYBE.  In Section ~\ref{sec:homlie}, we give some examples of Hom-Lie algebras and prove Theorem ~\ref{thm:homlie}.  We also show that $B_\alpha$ \eqref{eq:Balpha} is invertible, provided that $\alpha$ is invertible (Corollary ~\ref{cor:Balpha}).  In Sections ~\ref{sec:qtb} and ~\ref{sec:dqtb}, we recall the definitions of (dual) quasi-triangular bialgebras and prove Theorems ~\ref{thm:qtb} and ~\ref{thm:dqtb}.  In Section ~\ref{sec:braid}, we prove Theorem ~\ref{thm:braid} and illustrate it with some Hom-Lie deformations of $sl(2)$ (Example ~\ref{ex:sl2}).

\section{Hom-Yang-Baxter equation}
\label{sec:hybe}

Before we discuss some basic solutions of the HYBE, let us fix some notations.

\subsection{Conventions and notations}
\label{subsec:notations}

Throughout the rest of this paper, $\bk$ denotes a field of characteristic $0$.  Vector spaces, tensor products, and linearity are all meant over $\bk$, unless otherwise specified.

Given two vector spaces $V$ and $W$, denote by $\tau = \tau_{VW} \colon V \otimes W \to W \otimes V$ the twist isomorphism, i.e., $\tau(v \otimes w) = w \otimes v$.  Denote by $\Hom(V,W)$ the vector space of linear maps from $V$ to $W$.

For a coalgebra $C$ with comultiplication $\Delta \colon C \to C \otimes C$, we use Sweedler's notation $\Delta(x) = \sum_{(x)} x'\otimes x''$ \cite{sweedler}.  Suppose, in addition, that $M$ is a $C$-comodule with structure map $\rho \colon M \to C \otimes M$. For $u \in M$, we write $\rho(u) = \sum_{(u)} u_C \otimes u_M$.

\subsection{Hom-modules}

By a \textbf{Hom-module}, we mean a pair $(V, \alpha)$ in which $V$ is a vector space and $\alpha \colon V \to V$ is a linear map.
A morphism $(V, \alpha_V) \to (W, \alpha_W)$ of Hom-modules is a linear map $f \colon V \to W$ such that $\alpha_W \circ f = f \circ \alpha_V$.  When there is no danger of confusion, we will denote a Hom-module $(V,\alpha)$ simply by $V$.

The tensor product of the Hom-modules $(V, \alpha_V)$ and $(W, \alpha_W)$ consists of the vector space $V \otimes W$ and the linear self-map $\alpha_V \otimes \alpha_W$.

\subsection{Hom-Yang-Baxter equation}


Let $(M,\alpha)$ be a Hom-module, and let $B \colon M \otimes M \to M \otimes M$ be a morphism of Hom-modules, i.e., $B \circ \alpha^{\otimes 2} = \alpha^{\otimes 2} \circ B$.  Then the \textbf{Hom-Yang-Baxter Equation} (HYBE) for the Hom-module $(M,\alpha)$ is defined in \eqref{eq:HYBE}.


In the rest of this section, we give a few basic examples of solutions of the HYBE.

\begin{example}
If $B \colon V \otimes V \to V \otimes V$ is an $R$-matrix \eqref{eq:YBE}, then $B$ is also a solution of the HYBE for the Hom-module $(V,Id_V)$.\qed
\end{example}

\begin{example}
Let $(M,\alpha)$ be a Hom-module.  Define a map $\tau_\alpha \colon M \otimes M \to M \otimes M$ by
\[
\tau_\alpha(u \otimes v) = \alpha(v) \otimes \alpha(u)
\]
for $u,v \in M$.  Then clearly $\tau_\alpha$ is a morphism of Hom-modules.  Moreover, $\tau_\alpha$ is a solution of the HYBE for $(M,\alpha)$.  Indeed, with $B = \tau_\alpha$, both sides of \eqref{eq:HYBE}, when applied to $u \otimes v \otimes w \in M^{\otimes 3}$, are equal to $\alpha^3(w) \otimes \alpha^3(v) \otimes \alpha^3(u)$.\qed
\end{example}

\begin{proposition}
\label{prop:Binverse}
Let $B$ be a solution of the HYBE for the Hom-module $(M,\alpha)$.
\begin{enumerate}
\item
If $\lambda \in \bk$, then $\lambda B$ is also a solution of the HYBE for $(M,\alpha)$.
\item
If both $\alpha$ and $B$ are invertible, then $B^{-1}$ is a solution of the HYBE for $(M,\alpha^{-1})$.
\end{enumerate}
\end{proposition}

\begin{proof}
The first assertion follows from $\lambda B \circ \alpha^{\otimes 2} = \lambda(B \circ \alpha^{\otimes 2})$, $\alpha^{\otimes 2} \circ \lambda B = \lambda(\alpha^{\otimes 2} \circ B)$, $\lambda B \otimes \alpha = \lambda(B \otimes \alpha)$, and $\alpha \otimes \lambda B = \lambda(\alpha \otimes B)$.

The second assertion follows from $B^{-1} \circ (\alpha^{-1})^{\otimes 2} = (\alpha^{\otimes 2} \circ B)^{-1}$, $(\alpha^{-1})^{\otimes 2} \circ B^{-1} = (B \circ \alpha^{\otimes 2})^{-1}$, $B^{-1} \otimes \alpha^{-1} = (B \otimes \alpha)^{-1}$, and $\alpha^{-1} \otimes B^{-1} = (\alpha \otimes B)^{-1}$.
\end{proof}

\section{Solutions of the HYBE from Hom-Lie algebras}
\label{sec:homlie}

The purpose of this section is to prove Theorem ~\ref{thm:homlie}, which says that every Hom-Lie algebra gives a solution of the HYBE \eqref{eq:HYBE}.  We will also observe that, in the setting of Theorem ~\ref{thm:homlie}, if $\alpha$ is invertible, then so is $B_\alpha$ (Corollary ~\ref{cor:Balpha}).

We already defined Hom-Lie algebras in the Introduction.  Before we prove Theorem ~\ref{thm:homlie}, let us first discuss two classes of examples.


\begin{example}[\textbf{Hom-Lie algebras from Hom-associative algebras}]
The fact that associative algebras give rise to Lie algebras via the commutator bracket has a Hom-algebra counterpart.  By a \textbf{Hom-associative algebra}, we mean a triple $(A,\mu,\alpha)$, in which $(A,\alpha)$ is a Hom-module and $\mu \colon A \otimes A \to A$ is a bilinear map.  This data is required to satisfy two conditions: \emph{multiplicativity}, $\alpha(\mu(x,y)) = \mu(\alpha(x),\alpha(y))$, and \emph{Hom-associativity},
\[
\mu(\alpha(x),\mu(y,z)) = \mu(\mu(x,y),\alpha(z)),
\]
for $x,y,z \in A$.  An associative algebra is an example of a Hom-associative algebra in which $\alpha = Id$.  Hom-associative algebras (without multiplicativity) were introduced in \cite{ms}, and they play the roles of associative algebras in the Hom-algebra setting.  Indeed, given a Hom-associative algebra $(A,\mu,\alpha)$, we obtain a Hom-Lie algebra $(A,[-,-],\alpha)$ \cite{ms} in which
\[
[x,y] = \mu(x,y) - \mu(y,x)
\]
for $x,y \in A$.  It is clear that this bracket is skew-symmetric and multiplicative. The Hom-Jacobi identity \eqref{eq:hom-jacobi} can be proved by a direct computation using the Hom-associativity of $\mu$.

Although it is not needed in this paper, we point out that, similar to the universal enveloping algebra of a Lie algebra, there is a universal enveloping Hom-associative algebra functor $\Ubar$ from Hom-Lie algebras to Hom-associative algebras \cite{yau,yau3}.  The ordinary enveloping algebra of a Lie algebra is a bialgebra.  Similarly, one can define a \emph{Hom-bialgebra} by dualizing and extending the definition of a Hom-associative algebra.  It is shown in \cite[Section 3]{yau3} that, for a Hom-Lie algebra $L$, its universal enveloping Hom-associative algebra $\Ubar(L)$ has the structure of a Hom-bialgebra.  Besides the sources cited above, other papers that discuss Hom-Lie algebras and related Hom-algebras include \cite{hls,ms2,ms3,ms4,yau2}.\qed
\end{example}

\begin{example}[\textbf{Hom-Lie algebras as deformations of Lie algebras}]
\label{ex:homlie}
Another systematic way to obtain Hom-Lie algebras is by deforming Lie algebras along endomorphisms.  Let $(L,[-,-])$ be a Lie algebra, and let $\alpha \colon L \to L$ be a Lie algebra endomorphism.  Define a new bracket $[-,-]_\alpha$ on $L$ by setting
\[
[x,y]_\alpha = \alpha[x,y].
\]
Then a direct calculation shows that $L_\alpha = (L,[-,-]_\alpha,\alpha)$ is a Hom-Lie algebra \cite[Theorem 3.3]{yau2}.  Using this construction, one can easily obtain many examples of Hom-Lie algebras.  The reader is referred to \cite[Section 3]{yau2} for examples of Hom-Lie deformations of $sl(n)$, the Heisenberg algebra, Lie algebras associated to Lie groups, and the Witt algebra.

We point out that the procedure described in the previous paragraph can be applied to other types of algebras as well.  Indeed, one can deform an associative algebra along an endomorphism and obtain a Hom-associative algebra \cite[Theorem 2.5]{yau2}.  Dualizing and extending this procedure \cite{ms4,yau3}, one can obtain Hom-coalgebras and Hom-bialgebras by deforming coalgebras and bialgebras, respectively, along endomorphisms.\qed
\end{example}

After the above discussion of examples of Hom-Lie algebras, we now proceed to prove Theorem ~\ref{thm:homlie}.

\begin{proof}[Proof of Theorem ~\ref{thm:homlie}]
First, it is clear that $B_\alpha \circ \alpha^{\otimes 2} = \alpha^{\otimes 2} \circ B_\alpha$.  To prove that $B_\alpha$ satisfies the HYBE \eqref{eq:HYBE}, consider a typical generator $\gamma = (a,x) \otimes (b,y) \otimes (c,z)$ in $L'^{\otimes 3}$.  A direct computation gives:
\begin{equation}
\label{AF}
\begin{split}
(\alpha & \otimes B_\alpha) \circ (B_\alpha \otimes \alpha) \circ (\alpha \otimes B_\alpha)(\gamma)\\
& = (c,\alpha^3(z)) \otimes (b,\alpha^3(y)) \otimes (a,\alpha^3(x)) + (c,\alpha^3(z)) \otimes (1,0) \otimes (0,[\alpha^2(x),\alpha^2(y)])\\
& + (1,0) \otimes (b,\alpha^3(y)) \otimes (0,\alpha [\alpha(x),\alpha(z)]) + (1,0) \otimes (1,0) \otimes (0,[[\alpha(x),\alpha(z)],\alpha^2(y)])\\
& + (1,0) \otimes (0,\alpha^2[y,z]) \otimes (a,\alpha^3(x)) + (1,0) \otimes (1,0) \otimes (0,[\alpha^2(x),\alpha [y,z]]).
\end{split}
\end{equation}
Likewise, we have:
\begin{equation}
\label{AE'}
\begin{split}
(B_\alpha & \otimes \alpha) \circ (\alpha \otimes B_\alpha) \circ (B_\alpha \otimes \alpha)(\gamma)\\
& = (c,\alpha^3(z)) \otimes (b,\alpha^3(y)) \otimes (a,\alpha^3(x)) + (1,0) \otimes (0,[\alpha^2(y),\alpha^2(z)]) \otimes (a,\alpha^3(x))\\
& + (1,0) \otimes (b,\alpha^3(y)) \otimes (0,\alpha[\alpha(x),\alpha(z)]) + (c,\alpha^3(z)) \otimes (1,0) \otimes (0,\alpha^2[x,y])\\
& + (1,0) \otimes (1,0) \otimes (0,\alpha[[x,y],\alpha(z)]).
\end{split}
\end{equation}
Using the multiplicativity of $\alpha$, four terms in \eqref{AF} (those not of the form $(1,0) \otimes (1,0) \otimes \cdots$) are equal to four corresponding terms in \eqref{AE'}.  Therefore, $B_\alpha$ satisfies the HYBE, provided
\begin{equation}
\label{eq:alphaHJ}
[[\alpha(x),\alpha(z)],\alpha^2(y)] + [\alpha^2(x),\alpha[y,z]] = \alpha[[x,y],\alpha(z)].
\end{equation}
Using the multiplicativity of $\alpha$ and the skew-symmetry of the bracket, the condition \eqref{eq:alphaHJ} can be rewritten as
\begin{equation}
\label{eq:alphaHJ'}
0 = \alpha\left([[x,y],\alpha(z)] + [[z,x],\alpha(y)] + [[y,z],\alpha(x)]\right).
\end{equation}
The condition \eqref{eq:alphaHJ'} is true because of the Hom-Jacobi identity \eqref{eq:hom-jacobi} in $L$.
\end{proof}

\begin{corollary}
\label{cor:Balpha}
With the same hypotheses as in Theorem ~\ref{thm:homlie}, assume in addition that $\alpha \colon L \to L$ is invertible.  Then $B_\alpha$ is also invertible, whose inverse is given by
\[
B_\alpha^{-1}((a,x) \otimes (b,y)) = (b,\alpha^{-1}(y)) \otimes (a,\alpha^{-1}(x)) + (0,\alpha^{-2}[x,y]) \otimes (1,0).
\]
Moreover, $B_\alpha^{-1}$ is a solution of the HYBE for $(L',\alpha^{-1})$, where $\alpha^{-1}(a,x) = (a,\alpha^{-1}(x))$.
\end{corollary}

\begin{proof}
A direct computation shows that the stated $B_\alpha^{-1}$ is indeed the two-sided inverse of $B_\alpha$ \eqref{eq:Balpha}.  The last assertion follows from the second part of Proposition ~\ref{prop:Binverse}.
\end{proof}

\section{Solutions of the HYBE from quasi-triangular bialgebras}
\label{sec:qtb}

The purpose of this section is to prove Theorem ~\ref{thm:qtb}.  Let us first recall some relevant definitions.

\subsection{Quasi-triangular bialgebras}

Let $H = (H,\mu,\eta,\Delta,\varepsilon)$ be a bialgebra, in which $\mu \colon H \otimes H \to H$ is the associative multiplication, $\eta \colon \bk \to H$ is the unit, $\Delta \colon H \to H \otimes H$ is the coassociative comultiplication, and $\varepsilon \colon H \to \bk$ is the counit.  A \textbf{quasi-triangular structure} on $H$ \cite{dri87,dri89,kassel,majid} is an invertible element $R \in H \otimes H$ such that
\begin{equation}
\label{eq:almostcocom}
(\tau \circ \Delta)(x) = R \Delta(x) R^{-1}
\end{equation}
for $x \in H$,
\begin{equation}
\label{eq:R}
(\Delta \otimes Id_H)(R) = R_{13}R_{23}, \quad \text{and} \quad
(Id_H \otimes \Delta)(R) = R_{13}R_{12}.
\end{equation}
Here, if $R = \sum_i s_i \otimes t_i$, then
\[
R_{12} = \sum_i s_i \otimes t_i \otimes 1,\quad
R_{13} = \sum_i s_i \otimes 1 \otimes t_i,\quad \text{and} \quad
R_{23} = \sum_i 1 \otimes s_i \otimes t_i.
\]
We call $(H,R)$ a \textbf{quasi-triangular bialgebra}, which is also known as a \textbf{braided bialgebra}.  The quasi-triangular structure $R$ is also known as a \textbf{universal $R$-matrix}.

The reader is referred to \cite{dri87,dri89,kassel,majid} for detailed discussions and examples of quasi-triangular bialgebras, many of which are quantum groups.  In a quasi-triangular bialgebra $(H,R)$, the quasi-triangular structure $R$ satisfies the \textbf{Quantum Yang-Baxter Equation} (QYBE)
\begin{equation}
\label{eq:QYBE}
R_{12}R_{13}R_{23} = R_{23}R_{13}R_{12}.
\end{equation}
See, e.g., \cite[Theorem VIII.2.4]{kassel} for the proof.  Using the above notations, the QYBE \eqref{eq:QYBE} can be rewritten as
\begin{equation}
\label{eq:QYBE'}
\sum_{i,j,k} s_ks_j \otimes t_ks_i \otimes t_jt_i = \sum_{i,j,k} s_js_i \otimes s_kt_i \otimes t_kt_j.
\end{equation}
Consider the permutation isomorphism $\sigma \colon H^{\otimes 3} \to H^{\otimes 3}$ defined as $\sigma(x \otimes y \otimes z) = z \otimes y \otimes x$.  Applying $\sigma$ to both sides of \eqref{eq:QYBE'} yields
\begin{equation}
\label{eq:QYBE''}
\sum_{i,j,k} t_jt_i \otimes t_ks_i \otimes s_ks_j  = \sum_{i,j,k} t_kt_j \otimes s_kt_i \otimes s_js_i.
\end{equation}
We will make use of \eqref{eq:QYBE''} below.

\begin{proof}[Proof of Theorem ~\ref{thm:qtb}]
Let $u,v$, and $w$ be generic elements in $M$.  Using the notations above, the map $B_R$ \eqref{eq:BR} can be written as
\[
B_R(u \otimes v) = \sum_i t_iv \otimes s_iu,
\]
where $R = \sum_i s_i \otimes t_i$.  Since $\alpha \colon M \to M$ is $H$-linear, it is easy to see that $\alpha^{\otimes 2} \circ B_R = B_R \circ \alpha^{\otimes 2}$.

To see that $B_R$ satisfies the HYBE \eqref{eq:HYBE}, we write $\gamma = u \otimes v \otimes w$.  Using the $H$-linearity of $\alpha$, a direct computation shows that
\begin{equation}
\label{eq:BR1}
((\alpha \otimes B_R) \circ (B_R \otimes \alpha) \circ (\alpha \otimes B_R))(\gamma) = \sum_{i,j,k} t_jt_i\alpha(w) \otimes t_ks_i\alpha(v) \otimes s_ks_j\alpha(u)
\end{equation}
and
\begin{equation}
\label{eq:BR2}
((B_R \otimes \alpha) \circ (\alpha \otimes B_R) \circ (B_R \otimes \alpha))(\gamma) = \sum_{i,j,k} t_kt_j\alpha(w) \otimes s_kt_i\alpha(v) \otimes s_js_i\alpha(u).
\end{equation}
It follows from \eqref{eq:QYBE''} that \eqref{eq:BR1} and \eqref{eq:BR2} are equal.  Thus, $B_R$ satisfies the HYBE.
\end{proof}

\section{Solutions of the HYBE from dual quasi-triangular bialgebras}
\label{sec:dqtb}

The purpose of this section is to prove Theorem ~\ref{thm:dqtb}.  Let us first recall some relevant definitions.

\subsection{Dual quasi-triangular bialgebras}

Let $H = (H,\mu,\eta,\Delta,\varepsilon)$ be a bialgebra.  A \textbf{dual quasi-triangular structure} on $H$ \cite{hay,kassel,lt,majid91,majid,sch} is a bilinear form $R \in \Hom(H \otimes H,\bk)$ such that the following three conditions are satisfied for $x,y,z \in H$:
\begin{enumerate}
\item
The bilinear form $R$ is invertible under the convolution product.  In other words, there exists a bilinear form $R^{-1} \in \Hom(H \otimes H,\bk)$ such that
\begin{equation}
\label{eq:Rinvertible}
\sum_{(x)(y)} R(x' \otimes y')R^{-1}(x'' \otimes y'') = \varepsilon(x)\varepsilon(y) = \sum_{(x)(y)} R^{-1}(x' \otimes y')R(x'' \otimes y'').
\end{equation}
\item
The multiplication $\mu$ is almost commutative in the sense that
\begin{equation}
\label{eq:quasicom}
\sum_{(x)(y)} y'x'R(x'' \otimes y'') = \sum_{(x)(y)} R(x' \otimes y')x''y''.
\end{equation}
\item
We have
\begin{subequations}
\label{eq:R123}
\begin{align}
R(xy \otimes z) & = \sum_{(z)} R(x \otimes z')R(y \otimes z'')\label{eq:R13R23},\\
R(x \otimes yz) & = \sum_{(x)} R(x' \otimes z)R(x'' \otimes y).\label{eq:R13R12}
\end{align}
\end{subequations}
\end{enumerate}
We call $(H,R)$ a \textbf{dual quasi-triangular bialgebra}, which is also known as a \textbf{cobraided bialgebra}.  The dual quasi-triangular structure $R$ is also called a \textbf{universal $R$-form}.

Note that \eqref{eq:Rinvertible} is dual to the invertibility of the quasi-triangular structure.  Likewise, \eqref{eq:quasicom} is dual to \eqref{eq:almostcocom}, and \eqref{eq:R123} is dual to \eqref{eq:R}.  The reader is referred to the references above for more detailed discussions and examples of dual quasi-triangular bialgebras.

In the context of Theorem ~\ref{thm:dqtb}, the coassociativity of the $H$-comodule structure map $\rho \colon M \to H \otimes M$ can be expressed as the equality
\begin{equation}
\label{eq:coass}
\sum_{(u)(u_M)} u_H \otimes (u_M)_H \otimes (u_M)_M = \sum_{(u)(u_H)} u_H' \otimes u_H'' \otimes u_M
\end{equation}
for $u \in M$.  The notations were specified in \S \ref{subsec:notations}.  Likewise, the $H$-linearity of the $H$-comodule morphism $\alpha \colon M \to M$ is equivalent to the equality
\begin{equation}
\label{eq:star}
\sum_{(u)} u_H \otimes \alpha(u_M) = \sum_{(\alpha(u))} \alpha(u)_H \otimes \alpha(u)_M
\end{equation}
for $u \in M$.

\begin{proof}[Proof of Theorem ~\ref{thm:dqtb}]
Let $u,v$, and $w$ be generic elements in $M$.  First, the following calculation shows that $B^R$ \eqref{eq:B^R} is a morphism of Hom-modules:
\[
\begin{split}
B^R(\alpha(u) \otimes \alpha(v))
& = \sum_{(\alpha(u))(\alpha(v))} R(\alpha(v)_H \otimes \alpha(u)_H) \alpha(v)_M \otimes \alpha(u)_M\\
& = \sum_{(u)(v)} R(v_H \otimes u_H) \alpha(v_M) \otimes \alpha(u_M) \quad \text{by \eqref{eq:star}}\\
& = \alpha^{\otimes 2}(B^R(u \otimes v)).
\end{split}
\]

To see that $B^R$ satisfies the HYBE \eqref{eq:HYBE} for $(M,\alpha)$, we write $\gamma = u \otimes v \otimes w$.  By Lemmas ~\ref{lem1:dqtb} and ~\ref{lem2:dqtb} below, both $((B^R \otimes \alpha) \circ (\alpha \otimes B^R) \circ (B^R \otimes \alpha))(\gamma)$ and $((\alpha \otimes B^R) \circ (B^R \otimes \alpha) \circ (\alpha \otimes B^R))(\gamma)$ are equal to
\begin{equation}
\label{dagga3}
\sum_{\substack{(v)(v_H)(\alpha(w))\\(\alpha(w)_H)(\alpha(u))}} R\left(v_H'\alpha(w)_H' \otimes \alpha(u)_H\right) R\left(\alpha(w)_H'' \otimes v_H''\right) \cdot \alpha(w)_M \otimes \alpha(v_M) \otimes \alpha(u)_M.
\end{equation}
Therefore, it suffices to prove the following two Lemmas.
\end{proof}

\begin{lemma}
\label{lem1:dqtb}
With the notations above, $((B^R \otimes \alpha) \circ (\alpha \otimes B^R) \circ (B^R \otimes \alpha))(\gamma)$ is equal to \eqref{dagga3}.
\end{lemma}

\begin{proof}
A direct computation shows that $((B^R \otimes \alpha) \circ (\alpha \otimes B^R) \circ (B^R \otimes \alpha))(\gamma)$ is equal to
\begin{multline}
\label{ddagga}
\sum_{\substack{(u)(v)(u_M)\\ (\alpha(w))(\alpha(v_M))(\alpha(w)_M)}} R(v_H \otimes u_H) R(\alpha(w)_H \otimes (u_M)_H) R((\alpha(w)_M)_H \otimes \alpha(v_M)_H) \cdot \\(\alpha(w)_M)_M \otimes \alpha(v_M)_M \otimes \alpha((u_M)_M).
\end{multline}
Using \eqref{eq:coass} on both $u$ and $\alpha(w)$, \eqref{ddagga} becomes
\begin{multline}
\label{ddagga2}
\sum_{\substack{(u)(u_H)(v)(\alpha(v_M))\\(\alpha(w))(\alpha(w)_H)}} R(v_H \otimes u_H') R(\alpha(w)_H' \otimes u_H'') R(\alpha(w)_H'' \otimes \alpha(v_M)_H) \cdot\\
\alpha(w)_M \otimes \alpha(v_M)_M \otimes \alpha(u_M).
\end{multline}
Next, using \eqref{eq:R13R23} with $z = u_H$, \eqref{ddagga2} becomes
\begin{equation}
\label{ddagga2'}
\sum_{\substack{(u)(v)(\alpha(v_M))\\(\alpha(w))(\alpha(w)_H)}} R(v_H\alpha(w)_H' \otimes u_H) R(\alpha(w)_H'' \otimes \alpha(v_M)_H) \cdot
\alpha(w)_M \otimes \alpha(v_M)_M \otimes \alpha(u_M).
\end{equation}
Now using \eqref{eq:star} on both $u$ and $v_M$, \eqref{ddagga2'} becomes
\begin{equation}
\label{ddagga2''}
\sum_{\substack{(v)(v_M)(\alpha(u))\\(\alpha(w))(\alpha(w)_H)}} R(v_H\alpha(w)_H' \otimes \alpha(u)_H) R(\alpha(w)_H'' \otimes (v_M)_H) \cdot \alpha(w)_M \otimes \alpha((v_M)_M) \otimes \alpha(u)_M.
\end{equation}
Finally, using \eqref{eq:coass} on $v$, one observes that \eqref{ddagga2''} is equal to \eqref{dagga3}.
\end{proof}

\begin{lemma}
\label{lem2:dqtb}
With the notations above, $((\alpha \otimes B^R) \circ (B^R \otimes \alpha) \circ (\alpha \otimes B^R))(\gamma)$ is equal to \eqref{dagga3}.
\end{lemma}

\begin{proof}
This is similar to the proof of Lemma ~\ref{lem1:dqtb}, so we will only give a sketch.  A direct computation shows that $((\alpha \otimes B^R) \circ (B^R \otimes \alpha) \circ (\alpha \otimes B^R))(\gamma)$ is equal to
\begin{multline}
\label{dagga}
\sum_{\substack{(v)(w)(\alpha(u))(w_M)\\(\alpha(u)_M)(\alpha(v_M))}} R(w_H \otimes v_H) R((w_M)_H \otimes \alpha(u)_H) R(\alpha(v_M)_H \otimes (\alpha(u)_M)_H) \cdot\\ \alpha((w_M)_M) \otimes \alpha(v_M)_M \otimes (\alpha(u)_M)_M.
\end{multline}
One shows that \eqref{dagga} is equal to \eqref{dagga3} by performing the following steps:
\begin{enumerate}
\item use \eqref{eq:coass} on $w$;
\item use \eqref{eq:R13R12} with $x = w_H$;
\item use \eqref{eq:star} on both $w$ and $v_M$;
\item use \eqref{eq:coass} on $v$;
\item use \eqref{eq:quasicom} with $x = v_H$ and $y = \alpha(u)_H$;
\item use \eqref{eq:R13R12} with $x = \alpha(w)_H$;
\item use \eqref{eq:R13R23} with $z = \alpha(u)_H$.
\end{enumerate}
\end{proof}

\section{Braid group representations from solutions of the HYBE}
\label{sec:braid}

\begin{proof}[Proof of Theorem ~\ref{thm:braid}]
It suffices to prove the first assertion about the braid relations \eqref{eq:braidrelations} for the $B_i$.  The second assertion about the existence and uniqueness of $\rho^B_n$ follows from the invertibility of the $B_i$ and the presentation of $\br$ in terms of the generators $\sigma_i$ ($1 \leq i \leq n-1$) and the braid relations.

Suppose, then, $j - i > 1$ for some $i$ and $j$ with $1 \leq i,j \leq n-1$.  Using the commutativity of $B$ with $\alpha^{\otimes 2}$, we have:
\[
\begin{split}
B_i \circ B_j 
&= (\alpha^2)^{\otimes (i-1)} \otimes (B \circ \alpha^{\otimes 2}) \otimes (\alpha^2)^{\otimes (j-i-2)} \otimes (\alpha^{\otimes 2} \circ B) \otimes (\alpha^2)^{\otimes (n-j-1)}\\
&= (\alpha^2)^{\otimes (i-1)} \otimes (\alpha^{\otimes 2} \circ B) \otimes (\alpha^2)^{\otimes (j-i-2)} \otimes (B \circ \alpha^{\otimes 2}) \otimes (\alpha^2)^{\otimes (n-j-1)} = B_j \circ B_i.
\end{split}
\]
Here $\alpha^2 = \alpha \circ \alpha$.  By the symmetry of $i$ and $j$, we conclude that $B_i \circ B_j = B_j \circ B_i$ if $|i - j| > 1$.  This proves the first braid relation in \eqref{eq:braidrelations} for the $B_i$.

For the other braid relation, we use the assumption that $B$ satisfies the HYBE \eqref{eq:HYBE} and compute as follows:
\[
\begin{split}
B_i \circ B_{i+1} \circ B_i 
&= (\alpha^3)^{\otimes (i-1)} \otimes [(B \otimes \alpha) \circ (\alpha \otimes B) \circ (B \otimes \alpha)] \otimes (\alpha^3)^{\otimes (n-i-2)}\\
&= (\alpha^3)^{\otimes (i-1)} \otimes [(\alpha \otimes B) \circ (B \otimes \alpha) \circ (\alpha \otimes B)] \otimes (\alpha^3)^{\otimes (n-i-2)}\\
&= B_{i+1} \circ B_i \circ B_{i+1}.
\end{split}
\]
We have shown that the $B_i$ satisfy the second braid relation in \eqref{eq:braidrelations} as well.
\end{proof}

\begin{example}[\textbf{Braid group representations from Hom-Lie $sl(2)$}]
\label{ex:sl2}
In this Example, we illustrate how to construct a parametric family of braid group representations from a Lie algebra, using Theorems ~\ref{thm:homlie} and ~\ref{thm:braid} and Corollary ~\ref{cor:Balpha}.

Let $sl(2)$ denote the Lie algebra of $2 \times 2$ matrices with trace $0$.  A standard linear basis of $sl(2)$ consists of the matrices
   \[
   h = \begin{pmatrix} 1 & 0 \\ 0 & -1 \end{pmatrix}, \quad
   e = \begin{pmatrix} 0 & 1 \\ 0 & 0 \end{pmatrix}, \quad \text{and} \quad
   f = \begin{pmatrix} 0 & 0 \\ 1 & 0 \end{pmatrix},
   \]
which satisfy the relations:
   \[
   \lbrack h,e \rbrack = 2e, \quad
   \lbrack h,f \rbrack = -2f, \quad \text{and} \quad
   \lbrack e,f \rbrack = h.
   \]
Each non-zero scalar $\lambda \in \bk$, the ground field, gives a Lie algebra morphism $\alpha_\lambda \colon sl(2) \to sl(2)$ defined by
   \[
   \alpha_\lambda(h) = h, \quad
   \alpha_\lambda(e) = \lambda e, \quad \text{and} \quad
   \alpha_\lambda(f) = \lambda^{-1}f
   \]
on the basis elements.  One can check directly that the new bracket $[-,-]_{\alpha_\lambda}$ on (the underlying vector space of) $sl(2)$ defined by
   \[
   \lbrack h,e \rbrack_{\alpha_\lambda} = 2\lambda e, \quad
   \lbrack h,f \rbrack_{\alpha_\lambda} = -2\lambda^{-1}f, \quad \text{and} \quad
   \lbrack e,f \rbrack_{\alpha_\lambda} = h
   \]
satisfies the Hom-Jacobi identity \eqref{eq:hom-jacobi} and that $\alpha_\lambda$ is multiplicative with respect to $[-,-]_{\alpha_\lambda}$.  Therefore, $sl(2)_\lambda = (sl(2), [-,-]_{\alpha_\lambda}, \alpha_\lambda)$ is a Hom-Lie algebra, as defined in the Introduction.  This is also an instance of Example ~\ref{ex:homlie}.

Applying Theorem ~\ref{thm:homlie} to the Hom-Lie algebra $sl(2)_\lambda$, we see that $B_{\alpha_\lambda}$ \eqref{eq:Balpha} is a solution of the HYBE for $(sl(2)' = \bk \oplus sl(2), \alpha_\lambda)$.  Moreover, $\alpha_\lambda$ is invertible with inverse $\alpha_\lambda^{-1} = \alpha_{\lambda^{-1}}$.  So Corollary ~\ref{cor:Balpha} tells us that $B_{\alpha_\lambda}$ is also invertible.  Thus, the linear maps $B_i$ ($1 \leq i \leq n-1$) on $sl(2)'^{\otimes n}$ \eqref{eq:Bi} defined by $B_{\alpha_\lambda}$ and $\alpha_\lambda$ are all invertible.  It follows from Theorem ~\ref{thm:braid} that there is a unique braid group representation $\rho^\lambda = \rho_n^{B_{\alpha_\lambda}} \colon \br \to \Aut(sl(2)'^{\otimes n})$ given by $\rho^\lambda(\sigma_i) = B_i$.

Starting from the Lie algebra $sl(2)$, we have constructed an infinite, $1$-parameter family $\{\rho^\lambda\} = \{\rho^\lambda \colon \lambda \in \bk \setminus \{0\}\}$ of representations of the braid group $\br$ on $sl(2)'^{\otimes n} = (\bk \oplus sl(2))^{\otimes n}$.  Moreover, suppose $\rho \colon \br \to \Aut(sl(2)'^{\otimes n})$ denotes the braid group representation associated to the $R$-matrix $B = B_{Id}$ \eqref{eq:Balpha} on $sl(2)' = \bk \oplus sl(2)$.  Taking $\lambda = 1$, we have $\alpha_1 = Id$, $B_{\alpha_1} = B_{Id} = B$, and $\rho^1 = \rho$.  We can, therefore, think of $\{\rho^\lambda\}$ as a $1$-parameter family of deformations of the usual braid group representation $\rho$ on $sl(2)'^{\otimes n}$.

The procedure above can be summarized as follows.  Take a Lie algebra $L$, and deform it into some family of Hom-Lie algebras $\{L_\alpha = (L,[-,-]_\alpha,\alpha)\}$ along some family of Lie algebra automorphisms $\alpha$ (Example ~\ref{ex:homlie}).  Apply Theorem ~\ref{thm:homlie} to each Hom-Lie algebra $L_\alpha$ to obtain a solution $B_\alpha$ of the HYBE on $(L' = \bk \oplus L, \alpha)$.  Then apply Corollary ~\ref{cor:Balpha} and Theorem ~\ref{thm:braid} to these $B_\alpha$ to obtain a family of representations of the braid group $\br$ on $L'^{\otimes n} = (\bk \oplus L)^{\otimes n}$.

This procedure can be applied to other Lie algebras to obtain parametric families of representations of the braid group.  One can use, for example, the ($1$-parameter or multi-parameter) Hom-Lie deformations of the Lie algebra $sl(n)$, the Heisenberg algebra, the Witt algebra, and matrix Lie algebras in \cite[Section 3]{yau2}.
\qed
\end{example}


\end{document}